%% file: icaart2017.tex
\documentclass[a4paper,twoside]{article}
\usepackage{calrsfs}
\DeclareMathAlphabet{\pazocal}{OMS}{zplm}{m}{n}
\DeclareMathAlphabet{\xitscal}{OMS}{xits}{m}{n}
\usepackage{epsfig}
\usepackage{subfigure}
\usepackage{calc}
\usepackage{amssymb}
\usepackage{amstext}
\usepackage{amsmath}
\usepackage{amsthm}
\usepackage{multicol}
\usepackage{pslatex}
\usepackage{apalike}
\usepackage[utf8]{inputenc}

\usepackage{nicefrac}
\usepackage[ruled,linesnumbered,noend,vlined]{algorithm2e}
\usepackage{dirtytalk}
\usepackage{mathtools}
\usepackage{todonotes}
\usepackage{tikz}

\usepackage{SCITEPRESS}     % Please add other packages that you may need BEFORE the SCITEPRESS.sty package.

\usetikzlibrary{shapes,snakes,positioning,backgrounds,calc}

\newcommand{\seq}[2]{\textsf{seq}_{#1}(#2)}
\newcommand{\infset}[2]{\textsf{inf}_{#1}(#2)}
\newcommand{\calI}[0]{{\pazocal{I}}}

\newcommand{\calB}[0]{{\xitscal{B}}}
\newcommand{\calC}[0]{{\xitscal{C}}}
\newcommand{\calZ}[0]{{\xitscal{Z}}}
\newcommand{\calA}[0]{{\xitscal{A}}}
\newcommand{\calH}[0]{{\xitscal{H}}}

\newcommand{\calN}[0]{{\xitscal{N}}}

\newtheorem{lemma}{Lemma}
\newtheorem{definition}{Definition}
\newtheorem{proposition}{Proposition}
\newtheorem{theorem}{Theorem}

\renewcommand{\epsilon}{\varepsilon}

\DeclarePairedDelimiter\ceil{\lceil}{\rceil}
\DeclarePairedDelimiter\floor{\lfloor}{\rfloor}

\DeclareMathOperator*{\argmax}{arg\,max}

\begin{document}

\title{Computing Maxmin Strategies in Extensive-Form Zero-Sum Games with Imperfect Recall}

\author{\authorname{Branislav Bošanský\sup{1}, Jiří Čermák\sup{1}, Karel Horák\sup{1} and Michal Pěchouček\sup{1}}
\affiliation{\sup{1}Department of Computer Science, Czech Technical University in Prague, Czech Republic}
\email{\{bosansky, cermak, horak, pechoucek\}@agents.fel.cvut.cz}
}

\keywords{Game Theory, Imperfect Recall, Maxmin Strategies.}

\abstract{\input{abstract.tex}}
\onecolumn \maketitle \normalsize \vfill
\input{intro_big.tex}
\input{background.tex}
\input{method.tex}

\input{bnb.tex}

\input{experiments.tex}

\input{conclusion.tex}
\section*{\uppercase{Acknowledgements}}
This research was supported by the Czech Science Foundation (grant no. 15-23235S) and by the Grant Agency of the Czech Technical University in Prague, grant No. SGS16/235/OHK3/3T/13. Computational resources were provided by the CESNET LM2015042 and the CERIT Scientific Cloud LM2015085, provided under the programme "Projects of Large Research, Development, and Innovations Infrastructures".

\bibliographystyle{apalike}
\bibliography{refs}
%\newpage
\section*{APPENDIX}
\input{appendix_new}

\end{document}

%% file: abstract.tex
%\bbtodo{Link to complexity results on arxive}
Extensive-form games with imperfect recall are an important game-theoretic model that allows a compact representation of strategies in dynamic strategic interactions.
Practical use of imperfect recall games is limited due to negative theoretical results: a Nash equilibrium does not have to exist, computing maxmin strategies is NP-hard, and they may require irrational numbers.
We present the first algorithm for approximating maxmin strategies in two-player zero-sum imperfect recall games without absentmindedness.
We modify the well-known sequence-form linear program to model strategies in imperfect recall games resulting in a bilinear program and use a recent technique to approximate the bilinear terms.
Our main algorithm is a branch-and-bound search that provably reaches the desired approximation after an exponential number of steps in the size of the game.
Experimental evaluation shows that the proposed algorithm can approximate maxmin strategies of randomly generated imperfect recall games of sizes beyond toy-problems within few minutes.

%% file: intro_big.tex
\section{\uppercase{Introduction}}
\noindent The extensive form is a well-known representation of dynamic strategic interactions that evolve in time.
Games in the extensive form (extensive-form games; EFGs) are visualized as game trees, where nodes correspond to states of the game and edges to actions executed by players.
This representation is general enough to model stochastic events and imperfect information when players are unable to distinguish among several states.
Recent years have seen advancements in algorithms for computing solution concepts in large zero-sum extensive-form games (e.g., solving heads-up limit texas hold'em poker~\cite{bowling2015heads}).

Most of the algorithms for finding optimal strategies in EFGs assume that players remember all information gained during the course of the game~\cite{zinkevich2008regret,Hoda2010,bosansky2014-jair}.
This assumption is known as \emph{perfect recall} and has a significant impact on theoretical properties of finding optimal strategies in EFGs.
Namely, there is an equivalence between two types of strategies in perfect recall games -- \emph{mixed strategies} (probability distributions over pure strategies\footnote{A pure strategy in an EFG is an assignment of an action to play in each decision point.}) and \emph{behavioral strategies} (probability distributions over actions in each decision point)~\cite{kuhn1953}.
This equivalence guarantees that a Nash equilibrium (NE) exists in behavioral strategies in perfect recall games (the proof of the existence of NE deals with mixed strategies only \cite{Nash1950}) and it is exploited by algorithms for computing a NE in zero-sum EFGs with perfect recall -- the well-known sequence-form linear program~\cite{koller1996,vonStengel96}.

The caveat of perfect recall is that remembering all information increases the number of decision points (and consequently the size of a behavioral strategy) exponentially with the number of moves in the game.
One possibility for tackling the size of perfect recall EFGs is to create an abstracted game where certain decision points are merged together, solve this abstracted game, and then translate the strategy from the abstracted game into the original game~(e.g., see \cite{gilpin2007,kroer2014extensive,kroer2014extensiveIR}).
However, devising abstracted games that have \emph{imperfect recall} is desirable due to the reduced size. One then must compute behavioral strategies in order to exploit the reduced size (mixed strategies already operate over an exponentially large set of pure strategies).

Solving imperfect recall games has several fundamental problems.
The best known game-theoretic solution concept, a Nash equilibrium (NE), does not have to exist even in zero-sum games (see~\cite{wichardt2008} for a simple example) and standard algorithms (e.g., a Counterfactual Regret Minimization (CFR)~\cite{zinkevich2008regret}) can converge to incorrect strategies (see Example~1).
Therefore, we focus on finding a strategy that guarantees the best possible expected outcome for a player -- \emph{a maxmin strategy}.
However, computing a maxmin strategy is NP-hard and such strategies may require irrational numbers even when the input uses only rational numbers~\cite{koller1992}.

Existing works avoid these negative results by creating very specific abstracted games so that perfect recall algorithms are still applicable.
One example is a subset of imperfect recall games called \emph{(skewed) well-formed games}, motivated by the poker domain, in which the standard perfect-recall algorithms (e.g., CFR) are still guaranteed to find an approximate Nash behavioral strategy~\cite{lanctot2012,kroer2014extensiveIR}.
The restrictions on games to form (skewed) well-formed games are, however, rather strict and can prevent us from creating sufficiently small abstracted  games.
To fully explore the possibilities of exploiting the concept of abstractions and/or other compactly represented dynamic games (e.g., Multi-Agent Influence Diagrams~\cite{koller2003multi}), a new algorithm for solving imperfect recall games is required.

\subsection{Our Contribution}

\noindent We advance the state of the art and provide the first approximate algorithm for computing maxmin strategies in imperfect recall games (since maxmin strategies might require irrational numbers~\cite{koller1992}, finding exact maxmin has fundamental difficulties).
We assume imperfect recall games with no \emph{absentmindedness}, which means that each decision point in the game can be visited at most once during the course of the game and it is arguably a natural assumption in finite games (see, e.g., \cite{piccione1997} for a detailed discussion).
The main goal of our approach is to find behavioral strategies that maximize the expected outcome of player 1 against an opponent that minimizes the outcome.
We base our formulation on the sequence-form linear program for perfect recall games~\cite{koller1996,vonStengel96} and we extend it with bilinear constraints necessary for the correct representation of strategies of player 1 in imperfect recall games.
We approximate the bilinear terms using recent Multiparametric Disaggregation Technique (MDT)~\cite{Kolodziej2013} and provide a mixed-integer linear program (MILP) for approximating maxmin strategies.
Finally, we consider a linear relaxation of the MILP and propose a branch-and-bound algorithm that (1) repeatedly solves this linear relaxation and (2) tightens the constraints that approximate bilinear terms as well as relaxed binary variables from the MILP.
We show that the branch-and-bound algorithm ends after exponentially many steps while guaranteeing the desired precision.

Our algorithm approximates maxmin strategies for player 1 having generic imperfect recall without absentmindedness and we give two variants of the algorithm depending on the type of imperfect recall of the opponent.
If the opponent, player 2, has either a perfect recall or so-called \emph{A-loss recall}~\cite{kaneko1995,kline2002}, the linear program solved by the branch-and-bound algorithm has a polynomial size in the size of the game.
If player 2 has a generic imperfect recall without absentmindedness, the linear program solved by the branch-and-bound algorithm can be exponentially large. % and our algorithm can be further enhanced with constraint-generation techniques that iteratively adds best responses of player 2 to reduce the size of the linear program in practice.

We provide a short experimental evaluation to demonstrate that our algorithm can solve games far beyond the size of toy problems.
Randomly generated imperfect recall games with up to $5 \cdot 10^3$ states can be typically solved within few minutes.

All the technical proofs can be found in the appendix or in the full version of this paper.

%% file: background.tex
\section{\uppercase{Technical Preliminaries}}
\noindent Before describing our algorithm we define extensive-form games, different types of recall, and describe the approximation technique for the bilinear terms.

A two-player extensive-form game (EFG) is a tuple $G=(\calN,\calH,\calZ,\calA,u,\calC,\calI)$.
$\calN = \{1, 2\}$ is a set of players, by $i$ we refer to one of the players, and by $-i$ to his opponent. 
$\calH$ denotes a finite set of \emph{histories} of actions taken by all players and chance from the root of the game. 
Each history corresponds to a \emph{node} in the game tree; hence, we use terms history and node interchangeably.
We say that $h$ is a \emph{prefix} of~$h'$ ($h \sqsubseteq h'$) if $h$ lies on a path from the root of the game tree to $h'$.
$\calZ \subseteq \calH$  is the set of  \emph{terminal states} of the game. 
$\calA$ denotes the set of all actions.
An ordered list of all actions of player~$i$ from root to $h$ is referred to as a \emph{sequence}, $\sigma_i = \seq{i}{h}$, $\Sigma_i$ is a set of all sequences of~$i$. 
%Finally, playing action $a$ in node $h$ leads to node (history) $ha$.
For each $z \in \calZ$ we define a \emph{utility function} $u_i : \calZ \rightarrow \mathbb{R}$ for each player~$i$ ($u_i(z) = -u_{-i}(z)$ in zero-sum games). 
The chance player selects actions based on a fixed probability distribution known to all players. 
Function $\calC : \calH \rightarrow [0,1]$ is the probability of reaching $h$ due to chance.%; ${\cal C}(h)$ is the product of chance probabilities of all actions in history~$h$.

Imperfect observation of player $i$ is modeled via \emph{information sets} $\calI_i$ that form a partition over $h \in \calH$ where $i$ takes action.
Player $i$ cannot distinguish between nodes in any $I_i \in \calI_i$.
$\calA(I_i)$ denotes actions available in each $h \in I_i$.
The action $a$ uniquely identifies the information set where it is available.
We use $\seq{i}{I_i}$ as a set of all sequences of player~$i$ leading to $I_i$.
Finally, we use $\infset{i}{\sigma_i}$ to be a set of all information sets to which sequence $\sigma_i$ leads.

%There are number of different types of strategies in EFGs.
%We are interested in computing behavioral maxmin strategies, however, we define all standard representations.
%\emph{Pure strategies} $\Pi_i$ assign one action for each $I \in \calI_i$. 
%A~\emph{mixed strategy} $\delta_i \in \Delta_i$ is a probability distribution over $\Pi_i$.
A \emph{behavioral strategy} $\beta_i \in \calB_i$ is a probability distribution over actions in each information set $I \in \calI_i$.
We use $u_i(\beta) = u_i(\beta_i, \beta_{-i})$ for the expected outcome of the game for~$i$ when players follow $\beta$.
A \emph{best response} of player~$i$ against $\beta_{-i}$ is a strategy $\beta_{i}^{BR} \in BR_i(\beta_{-i})$, where  
$u_i(\beta_{i}^{BR}, \beta_{-i}) \geq u_i(\beta_i', \beta_{-i})$ for all $\beta_i' \in \calB_i$.
$\beta_i(I,a)$ is the probability of playing $a$ in $I$, $\beta(h)$ denotes the probability that $h$ is reached when both players play according to $\beta$ and due to chance.

We say that $\beta_i$ and $\beta_i'$ are \emph{realization equivalent} if for any $\beta_{-i}$, $\forall z \in \calZ\ \beta(z) = \beta'(z)$, where $\beta = (\beta_i, \beta_{-i})$ and $\beta' = (\beta'_i, \beta_{-i})$.

A \emph{maxmin strategy} $\beta^*_i$ is defined as $\beta^*_i = \argmax_{\beta_i \in \calB_i}\min_{\beta_{-i} \in \calB_{-i}}u_i(\beta_i, \beta_{-i})$. Note that when a Nash equilibrium in behavioral strategies exists in a two-player zero-sum imperfect recall game then $\beta^*_i $ is a Nash equilibrium strategy for $i$.

\subsection{Types of Recall}
\noindent We now briefly define types of recall in EFGs and state several lemmas and observations about characteristics of strategies in imperfect recall EFGs that are later exploited by our algorithm.

In \emph{perfect recall}, all players remember the history of their own actions and all information gained during the course of the game. 
As a consequence, all nodes in any information set~$I_i$ have the same sequence for player~$i$.
If the assumption of perfect recall does not hold, we talk about games with \emph{imperfect recall}.
In imperfect recall games, mixed and behavioral strategies are not comparable~\cite{kuhn1953}.
However, in games without \emph{absentmindedness} (AM) where each information set is encountered at most once during the course of the game, the following observation allow us to consider only pure best responses of the opponent when computing maxmin strategies:

\begin{lemma}
Let $G$ be an imperfect recall game without AM and $\beta_1$ strategy of player~1. 
There exists an \emph{ex ante} (i.e., when evaluating only the expected value of the strategy) pure behavioral best response of player~2.
\end{lemma}
The proof is in the full version of the paper.

This lemma is applied when a mathematical program for computing maxmin strategies is formulated -- strategies of player 2 can be considered as constraints using pure best responses.
Note that this is not true in general imperfect recall games -- in games with AM, an ex ante best response may need to be randomized (e.g., in the game with absentminded driver~\cite{piccione1997}).

A disadvantage of using pure best responses as constraints for the minimizing player is that there are exponentially many pure best responses in the size of the game.
In perfect recall games, this can be avoided by formulating best-response constraints such that the opponent is playing the best action in each information set.
However, this type of response, termed \emph{time consistent strategy}~\cite{kline2002}, does not have to be an ex ante best response in general imperfect recall games (see~\cite{kline2002} for an example).
A class of imperfect recall games where it is sufficient to consider only time consistent strategies when computing best responses was termed as \emph{A-loss} recall games~\cite{kaneko1995,kline2002}.
\begin{definition}
Player $i$ has \emph{A-loss recall} if and only if for every $I \in \calI_i$ and nodes $h,h' \in I$ it holds either (1) $\seq{i}{h} = \seq{i}{h'}$, or (2) $\exists I' \in \calI_i$ and two distinct actions $a,a' \in \calA_i(I'), a\neq a'$ such that $a \in \seq{i}{h} \wedge a' \in \seq{i}{h'}$.
\end{definition}
Condition (1) in the definition says that if player~$i$ has perfect recall then she also has A-loss recall.
Condition (2)  requires that each loss of memory of A-loss recall player can be traced back to some loss of memory of the player’s own previous actions.

The equivalence between time consistent strategies and ex ante best responses allows us to simplify the best responses of player 2 in case she has A-loss recall. %Not only does it guarantee that the games model interactions with reasonable properties, it is also interesting from the computational perspective, as it also allows us to ensure that player 2 plays the best response in a concise way in the mathematical programs.
Formally, it is sufficient to consider best responses that correspond to the best response in a coarsest perfect-recall refinement of the imperfect recall game when computing best response for a player with A-loss recall.
By a \emph{coarsest perfect recall refinement} of an imperfect recall game $G$ we define a perfect recall game $G'$ where we split the imperfect recall information sets to biggest subsets still fulfilling the perfect recall.

\begin{definition}
\emph{The coarsest perfect recall refinement} $G'$ of the imperfect recall game $G = \{\calN, \calH, \calZ, \calA, u, \calC, \calI\}$ is a tuple $\{\calN, \calH, \calZ, \calA', u, \calC, \calI'\}$, where $\forall i \in \calN\ \forall I_i \in \calI_i$ $H(I_i)$ partitions information set $I_i$ such that $H(I_i) = \{H_1, ..., H_n\}$ is a disjoint partition of all $h \in I_i$, where $\bigcup_{j = 1}^n H_j = I_i$ and $\forall H_j \in H(I_i)\ \forall h_k, h_l \in H_j: \seq{i}{h_k} = \seq{i}{h_l}$ and $\forall h_k \in H_k, h_l \in H_l: H_k \cap H_l = \emptyset \Rightarrow \seq{i}{h_k} \neq \seq{i}{h_l}$. Each set from $H(I_i)$ corresponds to an information set $I_i' \in \calI_i'$.
Moreover, $\calA'$ is a modification of $\calA$ guaranteeing $\forall I \in \calI'\ \forall h_k, h_l \in I\ \calA'(h_k) = \calA'(h_l)$, while for all distinct $I^k, I^l \in \calI'\ \calA(I^k) \neq \calA(I^l)$.
\end{definition}

Note that we can restrict the coarsest perfect recall refinement only for $i$ by splitting only information sets of $i$ (information sets of $-i$ remain unchanged).
Finally, we assume that there is a mapping between actions from the coarsest perfect recall refinement $\calA'$ and actions in the original game $\calA$ so that we can identify to which actions from $\calA'$ an original action $a \in \calA$ maps. 
We assume this mapping to be implicit since it is clear from the context.

\begin{lemma}\label{lem:pureBR}
Let $G$ be an imperfect recall game where player 2 has A-loss recall and $\beta_1$ is a strategy of player 1, and let $G'$ be the coarsest perfect recall refinement of $G$ for player $2$.
Let $\beta'_2$ be a pure best response in $G'$ and let $\beta_2$ be a realization equivalent behavioral strategy in $G$, then $\beta_2$ is a pure best response to  $\beta_1$ in $G$.
\end{lemma}
The proof is in the full version of the paper.

Note that the $NP$-hardness proof of computing maxmin strategies due to Koller~\cite{koller1992} still applies, since we assume the maximizing player to have generic imperfect recall and the reduction provided by Koller results in a game where the maximizing player has generic imperfect recall while the minimizing player has perfect recall, which is a special case of both settings assumed in our paper.

Finally, let us show that CFR cannot be applied in these settings.
This is caused by the fact that CFR iteratively minimizes per information set regret terms (counterfactual regrets). Since in perfect recall games the sum of counterfactual regrets provides an upper bound on the external regret, this minimization is guaranteed to converge to a strategy profile with 0 external regret. In imperfect recall games, however, the sum of counterfactual regrets no longer forms an upper bound on the external regret \cite{lanctot2012}, and the minimization of these regret terms can, therefore, lead to a strategy profile with a non-zero external regret.

\noindent\emph{Example 1:} Consider the A-loss recall game in Figure \ref{fig:cfr_counter}. 
When setting the $x > 2$, one of the strategy profiles with zero counterfactual regret (and therefore a profile to which CFR can converge) is mixing uniformly between both $a$, $b$ and $g$, $h$, while player 2 plays $d$, $e$ deterministically.
By setting the utility $x$ as some large number, this strategy profile can have expected utility arbitrarily worse than the maxmin value $-1$. The reason is the presence of the conflicting outcomes for some action in an imperfect recall information set that cannot be generally avoided, or easily detected in imperfect recall games.
\begin{figure}
\centering
\includegraphics[width=3.5cm]{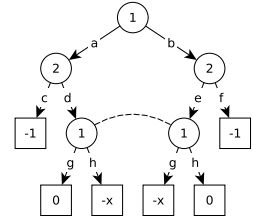}
\caption{An A-loss recall game where CFR finds a strategy with the expected utility arbitrarily distant from the maxmin value.}
\label{fig:cfr_counter}
\end{figure}
%Due to the space constraints, the proof is in the appendix.

\subsection{Approximating Bilinear Terms}\label{sec:MDT}
\noindent The final technical tool that we use in our algorithm is the approximation of bilinear terms by Multiparametric Disaggregation Technique (MDT)~\cite{Kolodziej2013} for approximating bilinear constraints. 
The main idea of the approximation is to use a digit-wise discretization of one of the variables from a bilinear term. 
The main advantage of this approximation is a low number of newly introduced integer variables and an experimentally confirmed speed-up over the standard technique of piecewise McCormick envelopes~\cite{Kolodziej2013}. 

{\small
\begin{subequations}\label{eq:mdt}
\begin{align}
\smash[t]{\sum_{k=0}^9} w_{k,\ell} & = 1&&\ell \in \mathbb{Z}\label{eq:mdt:1}\\
w_{k,\ell}  \in &\{0,1\}&& \\
\smash[t]{\sum_{\ell \in \mathbb{Z}}\sum_{k=0}^9} 10^\ell \cdot k \cdot w_{k,\ell} & = b && \label{eq:mdt:2}\\
c^L \cdot w_{k,\ell} \leq \hat{c}_{k,\ell} \leq c^U & \cdot w_{k,\ell} &&\forall \ell \in \mathbb{Z}, \forall k \in 0..9 \label{eq:mdt:3}\\
\smash[b]{\sum_{k=0}^9} \hat{c}_{k,\ell} & = c &&\forall  \ell \in \mathbb{Z} \label{eq:mdt:4}\\
\sum_{\ell \in \mathbb{Z}}\sum_{k=0}^9 10^\ell\cdot k \cdot \hat{c}_{k,\ell} & = a &&\label{eq:mdt:5}
\end{align}
\end{subequations}}

Let $a = bc$ be a bilinear term. 
MDT discretizes variable $b$ and introduces new binary variables $w_{k,l}$ that indicate whether the digit on $\ell$-th position is $k$.
Constraint~(\ref{eq:mdt:1}) ensures that for each position $\ell$ there is exactly one digit chosen.
All digits must sum to $b$ (Constraint~(\ref{eq:mdt:2})).
Next, we introduce variables $\hat{c}_{k,\ell}$ that are equal to $c$ for such $k$ and $\ell$ where $w_{k,l} = 1$, and $\hat{c}_{k,\ell} = 0$ otherwise. 
$c^L$ and $c^U$ are bounds on the value of variable $c$.
The value of $a$ is given by Constraint~(\ref{eq:mdt:5}).

This is an exact formulation that requires infinite sums and an infinite number of constraints. 
However, by restricting the set of all possible positions $\ell$ to a finite set $\{P_L, \ldots, P_U\}$ we get a lower bound approximation. 
Following the approach in~\cite{Kolodziej2013} we can extend the lower bound formulation to compute an upper bound:

{\small
\begin{subequations}\label{eq:mdtub}
\begin{flalign}
\textrm{Constraints}\;(\ref{eq:mdt:1}),(\ref{eq:mdt:3}),(\ref{eq:mdt:4}) \nonumber\\
\smash[b]{\sum_{\ell \in \{P_L, \ldots, P_U\}}\sum_{k=0}^9} 10^\ell \cdot k \cdot w_{k,\ell} + \Delta b & = b \label{eq:mdtub:1}\\
0 \leq \Delta b & \leq 10^{P_L} \label{eq:mdtub:2}\\
%a = \sum_{\ell \in \mathbb{Z}}\sum_{k=0}^9 10^\ell\cdot k \cdot \hat{c}_{k,\ell} + \sum_{k=0}^1 10^{P_L} \cdot k \cdot \tilde{c}_k \\
%c^L \cdot \tilde{w}_k \leq \tilde{c}_k \leq c^U \cdot \tilde{w}_k \\
%c = \sum_{k=0}^1 \tilde{c}_k \\
%\sum_{k=0}^1 \tilde{w}_k = 1 \\
%0 \leq \tilde{w}_k \leq 1
\smash[tb]{\sum_{\ell \in \{P_L, \ldots, P_U\}}\sum_{k=0}^9} 10^\ell\cdot k \cdot \hat{c}_{k,\ell} + \Delta a & = a \\
c^L \cdot \Delta b \leq \Delta a & \leq c^U \cdot \Delta b \label{eq:mdtub:3}\\
\left(c - c^U\right)\cdot 10^{P_L} + c^U \cdot \Delta b& \leq \Delta a  \label{eq:mdtub:4}\\
\left(c - c^L\right)\cdot 10^{P_L} + c^L \cdot \Delta b & \geq \Delta a \label{eq:mdtub:5}
\end{flalign}
\end{subequations}}

Here, $\Delta b$ is assigned to every discretized variable $b$ allowing it to take up the value between two discretization points created due to the minimal value of $\ell$ (Constraints~(\ref{eq:mdtub:1})--(\ref{eq:mdtub:2})).
Similarly, we allow the product variable $a$ to be increased with variable $\Delta a = \Delta b \cdot c$.
To approximate the product of the delta variables, we use the McCormick envelope defined by Constraints~(\ref{eq:mdtub:3})--(\ref{eq:mdtub:5}).

%% file: method.tex
\section{\uppercase{Mathematical Programs for Approximating Maxmin Strategies}}\label{sec:method}

\noindent We now state the mathematical programs for approximating maxmin strategies.
The main idea is to add bilinear constraints into the sequence form LP to restrict to imperfect recall strategies.
We formulate an exact bilinear program, followed by the approximation of bilinear terms using MDT.

\subsection{Exact Bilinear Sequence Form Against A-loss Recall Opponent}
\label{sec:approx-bilinseq}
{\small
\begin{subequations}\label{eq:bilinseq}
\begin{align}
\max_{x,r,v}  \; &v(root, \emptyset) && \label{eq:bilinseq:obj}\\
s.t. \qquad r(\emptyset) & = 1 \label{eq:bilinseq:rinit}&&\\
0\leq r(\sigma) & \leq 1 &&\forall \sigma \in \Sigma_1 \label{eq:bilinseq:2}\\
\sum_{a \in \calA(I)}r(\sigma a) & = r(\sigma) &&\forall \sigma \in \Sigma_1, \forall I \in \infset{1}{\sigma_1}\label{eq:bilinseq:4}\\
\sum_{a \in \calA(I)} x(a) & = 1 &&\forall I \in \calI_1^{IR} \label{eq:bilinseq:x1}\\
0 \leq x(a) & \leq 1  &&\forall I \in \calI_1^{IR}, \forall a \in \calA(I)\label{eq:bilinseq:behavbounds}\\
r(\sigma)\cdot x(a)  & = r(\sigma a) &&\forall I \in \calI_1^{IR}, \forall a \in \calA(I),\nonumber\\&&& \label{eq:bilinseq:bilin}\forall \sigma \in \seq{1}{I}\\
\shortintertext{\begin{equation}\displaystyle\sum_{\sigma_1 \in \Sigma_1}\!\!\!g(\sigma_1,\sigma_2a)r_1(\sigma_1) +\!\!\!\!\!\!\!\!\! \sum_{I' \in \infset{2}{\sigma_2 a}}\!\!\!\!\!\!\!\!\!v(I'\!\!, \sigma_2a) \geq v(I, \sigma_2)\nonumber\end{equation}} 
\shortintertext{\centering\begin{equation}\forall I \in \calI_2, \forall \sigma_2 \in \seq{2}{I}, \forall a \in \calA(I)\label{eq:bilinseq:br}\end{equation}}
\end{align}
\end{subequations}
}
Constraints \eqref{eq:bilinseq:obj}--\eqref{eq:bilinseq:br} represent a bilinear reformulation of the sequence-form LP due to~\cite{vonStengel96} applied to the information set structure of an imperfect recall game $G$.
The objective of player 1 is to find a strategy that maximizes the expected utility of the game.
The strategy is represented by variables $r$ that assign probability to a sequence: $r(\sigma_1)$ is the probability that $\sigma_1 \in \Sigma_1$ will be played assuming that information sets, in which actions of sequence $\sigma_1$ are applicable, are reached due to player 2.
Probabilities $r$ must satisfy so-called network flow Constraints~(\ref{eq:bilinseq:2})--(\ref{eq:bilinseq:4}).
Finally, a strategy of player 1 is constrained by the best-responding opponent that selects an action minimizing the expected value in each $I \in \calI_{2}$ and for each $ \sigma_2 \in \seq{2}{I}$ that was used to reach $I$~(Constraint~(\ref{eq:bilinseq:br})).
These constraints ensure that the opponent plays the best response in the coarsest perfect recall refinement of $G$ and thus also in $G$ due to Lemma~\ref{lem:pureBR}.
The expected utility for each action is a sum of the expected utility values from immediately reachable information sets $I'$ and from immediately reachable leafs.
For the latter we use generalized utility function $g : \Sigma_1 \times \Sigma_2 \rightarrow \mathbb{R}$ defined as $g(\sigma_1, \sigma_2) = \sum_{z \in \calZ | \seq{1}{z} = \sigma_1 \wedge \seq{2}{z} = \sigma_2} u(z)\calC(z)$.

In imperfect recall games multiple $\sigma_i$ can lead to some imperfect recall information set $I_i \in \calI_i^{IR} \subseteq \calI_i$; hence, realization plans over sequences do not have to induce the same behavioral strategy for $I_i$.
Therefore, for each $I_i \in \calI_i^{IR}$ we define behavioral strategy $x(a)$ for each $a\in \calA(I_i)$ (Constraints~\eqref{eq:bilinseq:x1}--\eqref{eq:bilinseq:behavbounds}).
To ensure that the realization probabilities induce the same behavioral strategy in $I_i$, we add bilinear constraint $r(\sigma_i a) = x(a)\cdot r(\sigma_i)$ (Constraint~\eqref{eq:bilinseq:bilin}).

\subsubsection{Player 2 without A-Loss Recall.}
\noindent If player 2 does not have A-loss recall, the mathematical program must use each pure best response of player 2 $\pi_2 \in \Pi_2$ as a constraint as follows:
{\small
\begin{subequations}\label{eq:bilinseq_gir}
\begin{align}
\max_{x,r,v}  \; v(root) && \label{eq:bilinseq_gir:obj}\\
%s.t. \qquad r(\emptyset) & = 1 \label{eq:bilinseq_gir:rinit}&&\\
%0\leq r(\sigma_1) & \leq 1 &&\forall \sigma_1 \in \Sigma_1 \label{eq:bilinseq_gir:2}\\
%\sum_{a \in \calA(I_1)}r(\sigma_1a) & = r(\sigma_1) &&\forall \sigma_1 \in \Sigma_1, \forall I_1 \in \infset{1}{\sigma_1}\label{eq:bilinseq_gir:4}\\
%\sum_{a \in \calA(I_1)} x(I_1,a) & = 1 &&\forall I_1 \in \calI_1^{IR} \label{eq:bilinseq_gir:x1}\\
%0 \leq x(I_1,a) & \leq 1  &&\forall I_1 \in \calI_1^{IR}, \forall a \in \calA(I_1)\label{eq:bilinseq_gir:behavbounds} \\
%r(\sigma_1)\cdot x(I_1,a)  & = r(\sigma_1a) &&\forall I_1 \in \calI_1^{IR}, \forall \sigma_1 \in \seq{1}{I_1},\label{eq:bilinseq_gir:bilin}\\&&&\nonumber \forall a \in \calA(I_1)\\
\text{Constraints~\eqref{eq:bilinseq:rinit}--\eqref{eq:bilinseq:behavbounds}}\nonumber\\
\sum_{z \in \calZ \;|\; \pi_2(z) = 1}\!\!\!\!\!\!\!\!u(z)\calC(z)r(\seq{1}{z}) & \geq v(root) &&\forall \pi_2 \in \Pi_2\label{eq:bilinseq_gir:br}
\end{align}
\end{subequations}}
Since the modification does not change the parts of the program related to the approximation of strategies of player 1, all the following approximation methods, theorems, and the branch-and-bound algorithm are applicable for general imperfect recall games without absentminded players.

\subsection{Upper Bound MILP Approximation}
\noindent The upper bound formulation of the bilinear program follows the MDT example and uses ideas similar to Section~\ref{sec:MDT}.
In accord with the MDT, we represent every variable $x(a)$ using a finite number of digits.
Binary variables $w_{k,\ell}^{I_1,a}$ correspond to $w_{k,\ell}$ variables from the example shown in Section~\ref{sec:MDT} and are used for the digit-wise discretization of $x(a)$.
Finally, $\hat{r}(\sigma_1)_{k,\ell}^a$ correspond to $\hat{c}_{k,\ell}$ variables used to discretize the bilinear term $r(\sigma_1 a)$.
In order to allow variable $x(a)$ to attain an arbitrary value from $[0,1]$ interval using a finite number of digits of precision, we add an additional real variable $0 \leq \Delta x(a) \leq 10^{-P}$ that can span the gap between two adjacent discretization points.
Constraints~\eqref{eq:ub-milp:x} and \eqref{eq:ub-milp:deltax-bound} describe this loosening.
Variables $\Delta x(a)$ also have to be propagated to bilinear terms $r(\sigma_1) \cdot x(a)$ involving $x(a)$.
We cannot represent the product $\Delta r(\sigma_1 a) = r(\sigma_1) \cdot \Delta x(a)$ exactly and therefore we give bounds based on the McCormick envelope (Constraints~\eqref{eq:ub-milp:deltar-bound-1}--\eqref{eq:ub-milp:deltar-bound-2}).
{\small
\begin{subequations}\label{eq:ub-milp}
\begin{align}
\max_{x,r,v} & \; v(root,\emptyset) &&\label{eq:ub-milp:obj}\\
\shortintertext{s.t. \quad\textrm{Constraints\ \eqref{eq:bilinseq:rinit}} \textrm{ - \eqref{eq:bilinseq:behavbounds} , \ \eqref{eq:bilinseq:br}}}
w^{I,a}_{k,\ell} & \in \{0,1\} &&\forall I \in \calI_1^{IR}, \forall a \in \calA(I),\nonumber\\
&&& \forall k \in 0..9,\forall \ell \in -P..0\label{eq:ub-milp:wb}\\
\smash{\sum_{k=0}^9} w_{k,\ell}^{I,a} & = 1  &&\forall I \in \calI_1^{IR}, \forall a \in \calA(I), \nonumber\\
&&&\forall \ell \in -P..0\label{eq:ub-milp:w}\\
\shortintertext{\begin{equation}\smash{\sum_{\ell = -P}^{0}\sum_{k=0}^9} 10^\ell \cdot k \cdot w^{I,a}_{k,\ell} + \Delta x(a) = x(a)\nonumber\label{eq:ub-milp:x0}\end{equation}}\stepcounter{equation}
&&&\forall I \in \calI_1^{IR}, \forall a \in \calA(I) \tag{\ref{eq:ub-milp:x0}}\label{eq:ub-milp:x} \\
0 \leq \Delta x(a) & \leq 10^{-P} &&\forall I \in \calI_1^{IR}, \forall a \in \calA(I) \label{eq:ub-milp:deltax-bound}\\
0 \leq \hat{r}(\sigma)^a_{k,\ell} & \leq w^{I,a}_{k,\ell} &&\forall I \in \calI_1^{IR}, \forall a \in \calA(I),\label{eq:deltas:hatr-bound}\\&&&\nonumber\forall\sigma \in \seq{1}{I},\forall \ell \in -P..0  \\
\smash{\sum_{k=0}^9}\hat{r}(\sigma)^a_{k,\ell} & = r(\sigma) &&   \forall I \in \calI_1^{IR}, \forall\sigma \in \seq{1}{I}\nonumber\\&&&\forall \ell \in -P..0\label{eq:deltas:hatr-sum}\\
\shortintertext{\begin{equation}\smash{\sum_{\ell = -P}^{0}\sum_{k=0}^9} 10^\ell \cdot k \cdot \hat{r}(\sigma)^a_{k,\ell} + \Delta r(\sigma a) = r(\sigma a)\nonumber\end{equation}}
&&&\forall I \in \calI_1^{IR}, \forall a \in \calA(I),\\&&& \forall\sigma \in \seq{1}{I} \\
\shortintertext{$\left(r(\sigma)-1\right)\cdot 10^{-P} + \Delta x(a) \leq \Delta r(\sigma a)  \leq 10^{-P}\cdot r(\sigma)$}
&&&\forall I \in \calI_1^{IR}, \forall a \in \calA(I), \nonumber\\&&&\forall\sigma \in \seq{1}{I}\label{eq:ub-milp:deltar-bound-1}\\
0 \leq \Delta r(\sigma a) &\leq \Delta x(a) && \forall I \in \calI_1^{IR}, \forall\sigma \in \seq{1}{I}, \nonumber\\&&&\forall a \in \calA(I)\label{eq:ub-milp:deltar-bound-2}
\end{align}
\end{subequations}}
Due to this loose representation of $\Delta r(\sigma_1 a)$, the reformulation of bilinear terms is no longer exact and this MILP therefore yields an upper bound of the bilinear sequence form program~\eqref{eq:bilinseq}. Note that the MILP has both the number of variables and the number of constraints bounded by $O(|\calI|\cdot|\Sigma|\cdot P)$, where $|\Sigma|$ is the number of sequences of both players. The number of binary variables is equal to $10\cdot|\calI_1^{IR}|\cdot \calA_1^{max}\cdot P$, where $\calA_1^{max} = \max_{I \in \calI_1}|\calA_1(I)|$.

\subsection{Theoretical Analysis of the Upper Bound MILP}
\label{subsec:reconstruct}

\noindent The variables $\Delta x(a)$ and $\Delta r(\sigma)$ ensure that the optimal value of the MILP is an upper bound on the value of the bilinear program. The drawback is that the realization probabilities do not have to induce a valid strategy in the imperfect recall game $G$, i.e. if $\sigma_1, \sigma_2$ are two sequences leading to an imperfect recall information set $I_1 \in \calI_1^{IR}$ where action $a \in \calA(I_1)$ can be played, $r(\sigma_1 a) / r(\sigma_1)$ need not equal $r(\sigma_2 a) / r(\sigma_2)$. We will show that it is possible to create a valid strategy in $G$ which decreases the value by at most $\epsilon$, while deriving bound on this $\epsilon$.

Let $\beta^1(I_1),\ldots,\beta^k(I_1)$ be behavioral strategies in the imperfect recall information set $I_1 \in \calI_1^{IR}$  corresponding to realization probabilities of continuations of sequences $\sigma^1,\ldots,\sigma^k$ leading to $I_1$. These probability distributions can be obtained from the realization plan as $\beta^j(I_1,a) = r(\sigma^j a) / r(\sigma^j)$ for $\sigma^j \in \seq{1}{I_1}$ and $a \in \calA(I_1)$. We will omit the information set and use $\beta(a)$ whenever it is clear from the context. If the imperfect recall is violated in $I_1$, $\beta^j(a)$ may not be equal to $\beta^l(a)$ for some $j$, $l$ and action $a \in \calA(I_1)$.
\begin{proposition} It is always possible to construct a strategy $\beta(I_1)$ such that $\| \beta(I_1) - \beta^j(I_1) \|_1 \leq |\calA(I_1)| \cdot 10^{-P}$ for every $j$.\footnote{The L1 norm is taken as $\| x_1 - x_2 \|_1 = \sum_{a \in \calA(I_1)} | x_1(a) - x_2(a) |$}
\label{prop:recon}
\end{proposition}

We now connect the distance of a corrected strategy $\beta(I_1)$ from a set of behavioral strategies $\beta^1(I_1),\ldots,\beta^k(I_1)$ in $I_1 \in \calI_1^{IR}$  to the expected value of the strategy.

\begin{theorem}\label{th:milperror}
\label{lemma:alg-bound}
The error of the Upper Bound MILP is bounded by
{\small
\[ \epsilon = 10^{-P} \cdot d \cdot \calA_1^{max} \cdot \frac{v_{max}(\emptyset) - v_{min}(\emptyset)}{2} \text{,} \]}
where $d$ is the maximum number of player 1's imperfect recall information sets encountered on a path from the root to a terminal node, $\calA_1^{max} = \max_{I_1 \in \calI_1^{IR}} |\calA(I_1)|$ is the branching factor and $v_{min}(\emptyset)$, $v_{max}(\emptyset)$ are the lowest and highest utilities for player 1 in the whole game, respectively.
\end{theorem}
The idea of the proof is to bound the error in every $I_1 \in \calI_1^{IR}$ and propagate the error in a bottom-up fashion.

The error of Upper Bound MILP is bounded if the precision of all approximations of bilinear terms is $P$.
However, we can increase the precision for each term separately and thus design the following iterative algorithm (termed simply as \textsc{MILP} in the experiments): (1) start with the precision set to $0$ for all bilinear terms, (2) for each approximation of a bilinear term calculate the current error contribution (the difference between $\Delta r(\sigma_1 a)$ and $r(\sigma_1) \Delta x(a)$ multiplied by the expected utility) and increase the precision only for the term that contributes to the overall error the most.
Once the term with maximal error has already reached maximal precision $P$, Theorem~\ref{th:milperror} guarantees us that we are $\varepsilon$ close to the optimal solution.
Our algorithm in the following section simultaneously increases the precision for approximating bilinear terms together with searching for optimal values for binary variables.

%% file: bnb.tex
\section{\uppercase{Branch-and-Bound Algorithm}}\label{sec:bnb}
\noindent We now introduce a branch-and-bound (\textsc{BnB}) search for approximating maxmin strategies, that exploits the observation below and thus improves the performance compared to the previous MILP formulations. Additionally, we provide bounds on the overall runtime as a function of the desired precision.

The \textsc{BnB} algorithm works on the linear relaxation of the Upper Bound MILP and searches the \textsc{BnB} tree in the best first search manner.
In every node $n$, the algorithm solves the relaxed LP corresponding to node $n$, heuristically selects the information set $I$ and action $a$ contributing to the current approximation error the most, and creates successors of $n$ by restricting the probability $\beta_1(I, a)$ that $a$ is played in $I$. 
The algorithm adds new constraints to LP depending on the value of $\beta_1(I, a)$ by constraining (and/or introducing new) variables $w_{k, l}^{I_1,a}$ and creating successors of the \textsc{BnB} node in the search tree.
Note that $w_{k, l}^{I_1,a}$ variables correspond to binary variables in the MILP formulation.
This way, the algorithm simultaneously searches for the optimal approximation of bilinear terms as well as the assignment for binary variables.
The algorithm terminates when $\epsilon$-optimal strategy is found (using the difference of the global upper bound and the lower bound computed as described in Observation~1 below). 

\paragraph{Observation 1.} Even if the current assignment to variables $w_{k,\ell}^{I_1,a}$ is not feasible (they are not set to binary values), the realization plan produced is valid in the perfect recall refinement. We can fix it in the sense of Proposition~\ref{prop:recon} and use it to estimate the lower bound for the \textsc{BnB} subtree rooted in the current node without a complete assignment of all $w_{k,\ell}^{I_1,a}$ variables to either 0 or 1.

\begin{algorithm}[t]
{\small
\DontPrintSemicolon
\SetKwInOut{Input}{input}
\SetKwInOut{Output}{output}
\SetKwInOut{Parameter}{parameters}
\SetKwProg{BFunction}{function}{}{}
\SetKwFunction{CreateNode}{CreateNode}
\SetKwFunction{ExtractStrategy}{ReconstructStrategy}
\SetKwFunction{SelectAction}{SelectAction}
\SetKwFunction{SelectInformationSet}{SelectInformationSet}
\SetKwFunction{Solve}{Solve}
\SetKwFunction{BestResponse}{BestResponse}

\Input{Initial LP relaxation $LP_0$ of Upper Bound MILP using a $P=0$ discretization}
\Output{$\epsilon$-optimal strategy for a player having imperfect recall}
\Parameter{Bound on maximum error $\epsilon$, precision bounds for $x(a)$ variables $P_{max}(I_1,a)$}

\BlankLine

$\mathsf{fringe} \gets \lbrace$\CreateNode{$LP_0$}$\rbrace$ \;
$\mathsf{opt} \gets (\textsf{nil},-\infty,\infty)$ \;
\While{$\mathsf{fringe} \neq \varnothing$}{
  $(LP,lb,ub) \gets \argmax_{n \in \mathsf{fringe}} n.\mathsf{ub}$ \label{alg:bnb:node-selection:first} \;
  $\mathsf{fringe} \gets \mathsf{fringe} \setminus (LP,lb,ub)$ \label{alg:bnb:node-selection:last} \;
  \If{$\mathsf{opt}.\mathsf{lb} \geq n.\mathsf{ub}$}{
     \Return{\ExtractStrategy{$\mathsf{opt}$}} \label{alg:bnb:bestfound}
  }

    \If{$\mathsf{opt}.\mathsf{lb} < n.\mathsf{lb}$}{
      $\mathsf{opt} \gets n$ \label{alg:bnb:keepopt}
    }
    \BlankLine
    \If{$n.\mathsf{ub}-n.\mathsf{lb} \leq \epsilon$}{
      \Return{\ExtractStrategy{$\mathsf{opt}$}} \label{alg:bnb:eps}
    }
    \Else{
      $(I_1,a) \gets $ \SelectAction{$n$} \label{alg:bnb:selection} \;
      $P \gets $ number of digits of precision representing $x(a)$ in $LP$ \;
      $\mathsf{fringe} \gets \mathsf{fringe} \cup \lbrace$\CreateNode{$LP \cup \lbrace \sum_{k=0}^{\floor{\frac{a_{ub} + a_{lb}}{2}}_{-P}} w_{k,P}^{I_1,a} = 1 \rbrace$}$\rbrace$ \label{alg:bnb:left} \;
      $\mathsf{fringe} \gets \mathsf{fringe} \cup \lbrace$\CreateNode{$LP \cup \lbrace \sum_{k= \floor{\frac{a_{ub} + a_{lb}}{2}}_{-P}}^9 w_{k,P}^{I_1,a} = 1 \rbrace$}$\rbrace$ \label{alg:bnb:right} \;
      \If{\label{alg:bnb:Pmax}$P < P_{max}(I_1,a)$}{
        $\mathsf{fringe} \gets \mathsf{fringe} \cup \lbrace$\CreateNode{$LP \cup \lbrace w_{LP.x(a)_{-P},P}^{I_1,a}=1, $ introduce vars $ w_{0,P+1}^{I_1,a},\ldots,w_{9,P+1}^{I_1,a} $ and corresponding constraints from MDT $ \rbrace$}$\rbrace$ \label{alg:bnb:refine} \;
      }
    }
  
}
\Return{\ExtractStrategy{$\mathsf{opt}$}} \label{alg:bnb:reconstruct}

\BlankLine

\BFunction{\CreateNode{$LP$}}{
  $ub \gets $ \Solve{$LP$}\label{alg:bnb:cn:solve} \;
  $\beta_1 \gets $ \ExtractStrategy{$LP$}\label{alg:bnb:cn:extract}\;
  $lb \gets u_1(\beta_1, \BestResponse(\beta1))$ \label{alg:bnb:cn:lb}\;
  \Return{$(LP,lb,ub)$}
}
}

\caption{\textsc{BnB} algorithm}
\label{alg:bnb}
\end{algorithm}
Algorithm~\ref{alg:bnb} depicts the complete \textsc{BnB} algorithm.
It takes an LP relaxation of the Upper Bound MILP %(relaxing all the binary variables to belong to continuous interval $[0, 1]$) 
as its input. Initially, the maxmin strategy is approximated using 0 digits of precision after the decimal point (i.e. precision $P(I_1,a)=0$ for every variable $x(a)$). The algorithm maintains a set of active \textsc{BnB} nodes ($\textsf{fringe}$) and a candidate with the best guaranteed value $\textsf{opt}$.
The algorithm selects the node with the highest upper bound from $\textsf{fringe}$ at each iteration (lines~\ref{alg:bnb:node-selection:first}--\ref{alg:bnb:node-selection:last}). If there is no potential for improvement in the unexplored parts of the branch and bound tree, the current best solution is returned (line~\ref{alg:bnb:bestfound}) (upper bounds of the nodes added to the fringe in the future will never be higher than the current upper bound). Next, we check, whether the current solution has better lower bound than the current best, if yes we replace it (line \ref{alg:bnb:keepopt}). Since we always select the most promising node with respect to the upper bound, we are sure that if the lower bound and upper bound have distance at most $\epsilon$, we have found an $\epsilon$-optimal solution and we can terminate (line \ref{alg:bnb:eps}) (upper bounds of the nodes added to the fringe in the future will never be higher than the current upper bound). Otherwise, we heuristically select an action having the highest effect on the gap between the current upper and lower bound (line \ref{alg:bnb:selection}). We obtain the precision used to represent behavioral probability of this action. By default we add two successors of the current \textsc{BnB} node, each with one of the following constraints. $x(a) \leq \floor{\frac{a_{ub} + a_{lb}}{2}}_{-P}$ (line~\ref{alg:bnb:left}) and $x(a) \geq \floor{\frac{a_{ub} + a_{lb}}{2}}_{-P}$ (line~\ref{alg:bnb:right}), where $\floor{\cdot}_p$ is flooring of a number towards $p$ digits of precision and $a_{ub}$ and $a_{lb}$ are the lowest and highest allowed values of playing $x(a)$. This step performs binary halving restricting allowed values of $x(a)$ in current precision. Additionally, if the current precision is lower than the maximal precision $P_{max}(I_1, a)$ the gap between bounds may be caused by the lack of discretization points; hence, we add one more successor $\floor{v} \leq x(a) \leq \ceil{v}$, where $v$ is the current probability of playing $a$, while increasing the precision used for representing $x(a)$ (line~\ref{alg:bnb:refine}) (all the restriction to $x(a)$ in all 3 cases are done via $w_{k, l}^{I_1,a}$ variables). 

The function \texttt{CreateNode} computes the upper bound by solving the given LP (line \ref{alg:bnb:cn:solve}) and the lower bound, by using the heuristical construction of a valid strategy $\beta_1$ returning some convex combination of strategies found for $\sigma_1^k \in \seq{1}{I_1}$ (line \ref{alg:bnb:cn:extract}) and computing the expected value of $\beta_1$ against a best response to it.

Note that this algorithm allows us to plug-in custom heuristic for the reconstruction of strategies (line \ref{alg:bnb:cn:extract}) and for the action selection (line \ref{alg:bnb:selection}). 

\subsection{Theoretical Properties of the \textsc{BnB} Algorithm}
\label{sec:bnb:Pmax}
\noindent The \textsc{BnB} algorithm takes the error bound $\epsilon$ as an input. We provide a method for setting the $P_{max}(I_1,a)$ parameters appropriately to guarantee $\epsilon$-optimality. Finally, we provide a bound on the number of steps the algorithm needs to terminate.

\begin{theorem}
Let $P_{max}(I_1,a)$ be the maximum number of digits of precision used for representing variable $x(a)$ set as
{\small
\[ P_{max}(I_1,a) = \left\lceil \max_{h \in I_1} \log_{10} \frac{|\calA(I_1)| \cdot d \cdot v_{diff}(h)}{2 \epsilon} \right\rceil \text{,} \]}
where $v_{diff}(h) = v_{max}(h) - v_{min}(h)$.
With this setting Algorithm~\ref{alg:bnb} terminates and it is guaranteed to return an $\epsilon$-optimal strategy for player 1.
\label{th:prec}
\end{theorem}
The proof provides a bound on the error per node in every $I \in \calI^{IR}_i$ and propagates this bound through the game tree in a bottom up fashion.
\begin{theorem}
When using $P_{max}(I_1,a)$ from Theorem \ref{th:prec} for all $I_1 \in \calI_1$ and all $a \in \calA(I_1)$, the number of iterations of the \textsc{BnB} algorithm needed to find an $\epsilon$-optimal solution is in $O(3^{4S_1( \log_{10}(S_1 \cdot v_{diff}(\emptyset)) + 1)} 2^{-5S_1}\epsilon^{-5S_1})$, where $S_1 = |\calI_1|\calA_1^{max}$.
\end{theorem}
The proof derives a bound on the number of nodes in the \textsc{BnB} tree dependent on $\max_{I_1 \in \calI_i, a \in \calA_1(I_1)}P_{max}(I_1,a)$ and uses the formula from Theorem \ref{th:prec} to transform the bound to a function of $\epsilon$.

%% file: experiments.tex
\section{\uppercase{Experiments}}
\noindent We now demonstrate the practical aspects of our main branch-and-bound algorithm (\textsc{BnB}) described in Section~\ref{sec:bnb} and the iterative MILP variant described in Section~\ref{sec:method}.
We compare the algorithms on a set of random games where player 2 has A-loss recall.
Both algorithms were implemented in Java, each algorithm uses a single thread, 8 GB memory limit. We use IBM ILOG CPLEX 12.6 to solve all LPs/MILPs.

\subsection{Random Games}
\noindent Since there is no standardized collection of benchmark EFGs, we use randomly generated games in order to obtain statistically significant results.
We randomly generate a perfect recall game with varying branching factor and fixed depth of 6. To control the information set structure, we use observations assigned to every action -- for player $i$, nodes $h$ with the same observations generated by all actions in history belong to the same information set.
In order to obtain imperfect recall games with a non-trivial information set structure, we run a random abstraction algorithm which randomly merges information sets with the same action count, which do not cause absentmindedness. We generate a set of experimental instances by varying the branching factor. Such games are rather difficult to solve since (1) information sets can span multiple levels of the game tree (i.e., the nodes in an information set often have histories with differing sizes) and (2) actions can easily lead to leafs with very differing utility values. We always devise an abstraction which results to A-loss recall for the minimizing player.

\subsection{Results}
\begin{table}[t]
\centering
\caption{Average runtime and standard error in seconds needed to solve at least 100 different random games in every setting, while increasing the b.f.with fixed depth of 6.
\label{tab:results}}{
\begin{tabular}{l|c|c}
Algs $\setminus$ $b.f.$ & 3 & 4 \\
\hline \textsc{MILP} & $157.87 \pm  61.47$& $459.09 \pm  75.95$\\
\textsc{BnB} &$9.52 \pm  3.78$ & $184.50 \pm  48.22$ \\
\end{tabular}
}
\end{table}
\noindent Table \ref{tab:results} reports the average runtime and its standard error in seconds for both algorithms over at least 100 different random games with increasing branching factor and fixed depth of 6.
We limited the runtime for 1 instance to  2~hours (in the Table we report results only for instances where both algorithms finished under 2 hours).
The \textsc{BnB} algorithm was terminated 17 and 30 times after 2 hours for the reported settings respectively, while MILP algorithm was terminated 19 and 34 times.
Note that the random games form an unfavorable scenario for both algorithms since the construction of the abstraction is completely random, which makes conflicting behavior in merged information sets common. As we can see, however, even in these scenarios we are typically able to solve games with approximately $5\cdot10^3$ states  in several minutes.

%% file: conclusion.tex
\section{\uppercase{Conclusions}}
\noindent We provide the first algorithm for approximating maxmin strategies in imperfect recall zero-sum extensive-form games without absentmindedness. % against an opponent having A-loss recall.
We give a novel mathematical formulation for computing approximate strategies by means of a mixed-integer linear program (MILP) that uses recent methods in the approximation of bilinear terms.
Next, we use a linear relaxation of this MILP and introduce a branch-and-bound search (\textsc{BnB}) that simultaneously looks for the correct solution of binary variables and increases the precision for the approximation of bilinear terms.
We provide guarantees that both MILP and \textsc{BnB} find an approximate optimal solution.
Finally, we show that the algorithms are capable of solving games of sizes far beyond toy problems (up to $5 \cdot 10^3$ states) typically within few minutes in practice. 

Results presented in this paper provide the first baseline algorithms for the class of imperfect recall games that are of a great importance in solving large extensive-form games with perfect recall.
As such, our algorithms can be further extended to improve the current scalability, e.g., by employing incremental strategy generation methods.

%% file: appendix_new.tex
\addtocounter{proposition}{-1}
\addtocounter{lemma}{0}
\addtocounter{theorem}{-3}

\begin{proposition} It is always possible to construct a strategy $\beta(I_1)$ such that $\| \beta(I_1) - \beta^j(I_1) \|_1 \leq |\calA(I_1)| \cdot 10^{-P}$ for every $j$.
\end{proposition}
\begin{proof}
Probabilities of playing action $a$ in $\beta^1,\ldots,\beta^k$ can differ by at most $10^{-P}$, i.e. $| \beta^j(a) - \beta^l(a)| \leq 10^{-P}$ for every $j,l$ and action $a \in \calA(I_1)$. This is based on the MDT we used to discretize the bilinear program. Let us denote
{\small
\begin{eqnarray}
\underline{r}(\sigma_1 a) &=& \sum_{l=-P}^{0} \sum_{k=0}^9 10^\ell \cdot k \cdot \hat{r}(\sigma_1)_{k,\ell}^a \\
\underline{x}(I_1,a) &=& \sum_{l=-P}^{0} \sum_{k=0}^9 10^\ell \cdot k \cdot w_{k,\ell}^{I_1,a} \text{.}
\end{eqnarray}}
Constraints~\eqref{eq:deltas:hatr-bound} and \eqref{eq:deltas:hatr-sum} ensure that $\underline{r}(\sigma_1 a) = r(\sigma_1) \cdot \underline{x}(I_1, a)$. The only way how the imperfect recall can be violated is thus in the usage of $\Delta r(\sigma_1 a)$. We know however that $\Delta r(\sigma_1 a) \leq 10^{-P} \cdot r(\sigma_1)$ which ensures that the amount of imbalance in $\beta^1,\ldots,\beta^k$  is at most $10^{-P}$. Taking any of the behavioral strategies $\beta^1,\ldots,\beta^k$ as the corrected behavioral strategy $\beta(I_1)$ therefore satisfies $\| \beta(I_1) - \beta^j(I_1) \|_1 \leq \sum_{a \in \calA(I_1)} 10^{-P} = |\calA(I_1)| \cdot 10^{-P}$.
\end{proof}

We now provide the technical proof of Theorem \ref{lemma:alg-bound}.
First, we connect the distance of a corrected strategy $\beta(I_1)$ from a set of behavioral strategies $\beta^1(I_1),\ldots,\beta^k(I_1)$ in $I_1 \in \calI_1^{IR}$  to the expected value of the strategy. We start with bounding this error in a single node.

\begin{lemma}\label{lemma:infset-bound}
Let $h \in I_1$ be a history and $\beta^1$, $\beta^2$ be behavioral strategies (possibly prescribing different behavior in $I_1$) prescribing the same distribution over actions for all subsequent histories $h' \sqsupset h$. Let $v_{max}(h)$ and $v_{min}(h)$ be maximal and minimal utilities of player 1 in the subtree of $h$, respectively. Then the following holds:
{\small
\[ | v_{\beta^1}(h) - v_{\beta^2}(h) | \leq \frac{v_{diff}(h)}{2} \cdot \| \beta^1(I_1) - \beta^2(I_1) \|_1 \text{,} \]}
where $v_{\beta^j}(h)$ is the maxmin value $u(\beta^j,\beta_2^{BR})$ of strategy $\beta^j$ of player 1 given the play starts in $h$ and $v_{diff}(h) = v_{max}(h) - v_{min}(h)$.
\end{lemma}
\begin{proof}
Let us study strategies $\beta^1$ and $\beta^2$ in node $h$. Let us take $\beta^1(I_1)$ as a baseline and transform it towards $\beta^2(I_1)$. We can identify two subsets of $\calA(I_1)$ --- a set of actions $A^+$ where the probability of playing the action in $\beta^2$ was increased and $A^-$ where the probability was decreased. Let us denote
{\small
\[ C^\circ = \sum_{a \in A^\circ} | \beta^1(I_1,a) - \beta^2(I_1,a) |\quad \forall \circ \in \{+, -\}\text{.} \]}

We know that $C^+ = C^-$ (as strategies have to be probability distributions). Moreover we know that $\| \beta^1(I_1) - \beta^2(I_1) \|_1 = C^+ + C^-$. In the worst case, decreasing the probability of playing action $a \in A^-$ risks losing quantity proportional to the amount of this decrease multiplied by the highest utility in the subtree $v_{max}(h)$. For all actions $a \in A^-$ this loss is equal to
{\small \[ v_{max}(h) \cdot \sum_{a \in A^-} | \beta^1(I_1,a) - \beta^2(I_1,a) | = v_{max}(h) \cdot C^- \text{.} \]}
Similarly the increase of the probabilities of actions in $A^+$ can add in the worst case $v_{min}(h) \cdot C^+$ to the value of the strategy. This combined together yields
{\small
\begin{align*}
v_{\beta^2}(h) - &v_{\beta^1}(h) \geq -v_{max}(h) \cdot C^- + v_{min}(h) \cdot C^+ \\
&= [ -v_{max}(h) + v_{min}(h) ] \cdot C^+\\
&= \frac{-v_{max}(h) + v_{min}(h)}{2} \cdot 2C^+ \\
&= \frac{-v_{max}(h) + v_{min}(h)}{2} \cdot \| \beta^1(I_1) - \beta^2(I_1)\|_1 \text{.}
\end{align*}}
The strategies $\beta^1$, $\beta^2$ are interchangeable which results in the final bound on the difference of $v_{\beta^2}(h)$, $v_{\beta^1}(h)$.
\end{proof}

Now we are ready to bound the error in the whole game tree.
\begin{theorem}
The error of the Upper Bound MILP is bounded by
{\small
\[ \epsilon = 10^{-P} \cdot d \cdot \calA_1^{max} \cdot \frac{v_{max}(\emptyset) - v_{min}(\emptyset)}{2} \text{,} \]}
where $d$ is the maximum number of player 1's imperfect recall information sets encountered on a path from the root to a terminal node, $\calA_1^{max} = \max_{I_1 \in \calI_1^{IR}} |\calA(I_1)|$ is the branching factor and $v_{min}(\emptyset)$, $v_{max}(\emptyset)$ are the lowest and highest utilities for player 1 in the whole game, respectively.
\end{theorem}
\begin{proof}
We show an inductive way to compute the bound on the error and we show that the bound from Theorem~\ref{lemma:alg-bound} is its upper bound. Throughout the derivation we assume that the opponent plays to maximize the error bound. We proceed in a bottom-up fashion over the nodes in the game tree, computing the maximum loss $L(h)$ player~1 could have accumulated by correcting his behavioral strategy in the subtree of $h$, i.e.
{\small
\[ L(h) \geq u_h(\beta^0) - u_h(\beta^{IR}) \text{,} \]}
where $\beta^0$ is the (incorrect) behavioral strategy of player 1 acting according to the realization probabilities $r(\sigma)$ from the solution of the Upper Bound MILP, $\beta^{IR}$ is its corrected version and $u_h(\beta)$ is the expected utility of a play starting in history $h$ when player 1 plays according to $\beta$ and his opponent best responds (without knowing that the play starts in $h$). The proof follows in case to case manner.

(1) No corrections are made in subtrees of leafs $h$, thus the loss $L(h)=0$.

(2) The chance player selects one of the successor nodes based on the fixed probability distribution. The loss is then the expected loss over all child nodes $L(h) = \sum_{a \in \calA(h)} L(h \cdot a) \cdot \calC(h \cdot a) / \calC(h)$.
In the worst case, the chance player selects the child with the highest associated loss, therefore {\small $$L(h) \leq \max_{a \in \calA(h)} L(h \cdot a).$$}

(3) Player 2 wants to maximize player 1's loss. Therefore she selects such an action in her node $h$ that leads to a node with the highest loss,
$L(h) \leq \max_{a \in \calA(n)} L(h \cdot a)$.
This is a pessimistic estimate of the loss as she may not be able to pick the maximizing action in every state because of the imperfection of her information.

(4) If player 1's node $h$ is not a part of an imperfect recall information set, no corrective steps need to be taken. The expected loss at node $h$ is therefore $L(h) = \sum_{a \in \calA(h)} \beta^0(h,a) L(h \cdot a)$.
Once again in the worst case player 1's behavioral strategy $\beta^0(h)$ selects deterministically the child node with the highest associated loss, therefore $L(h) \leq \max_{a \in \calA(h)} L(h \cdot a)$.

(5) So far we have considered cases that only aggregate losses from child nodes. If player 1's node $h$ is part of an imperfect recall information set, the correction step may have to be taken.
Let $\beta^{-h}$ be a behavioral strategy where corrective steps have been taken for successors of $h$ and let us construct a strategy $\beta^h$ where the strategy was corrected in the whole subtree of $h$ (i.e. including $h$). Note that ultimately we want to construct strategy $\beta^\emptyset=\beta^{IR}$.

\noindent We know that values of children have been decreased by at most $\max_{a \in \calA(h)} L(h \cdot a)$, hence
$v_{\beta^0}(h) - v_{\beta^{-h}}(h) \leq \max_{a \in \calA(h)} L(h \cdot a)$.
Then we have to take the corrective step at the node $h$ and construct strategy $\beta^h$. From Lemma~\ref{lemma:infset-bound} and the observation about the maximum distance of behavioral strategies within a single imperfect recall information set $I_1$, we get:
{\small
\begin{align*}
v_{\beta^{-h}}(h) - v_{\beta^h}(h) &\leq \frac{v_{diff}(h)}{2} \cdot 10^{-P}|\calA_1(I_1)| \\
&\leq \frac{v_{diff}(\emptyset)}{2} \cdot 10^{-P}\calA_1^{max}
\end{align*}}
The loss in the subtree of $h$ is equal to $v_{\beta^0}(h) - v_{\beta^{-h}}(h)$ which is bounded by
{\small
\begin{align*}
L(&h)  =  v_{\beta^0}(h) - v_{\beta^h}(h) \\
&= \left[ v_{\beta^{-h}}(h) - v_{\beta^h}(h) \right] + \left[ v_{\beta^0}(h) - v_{\beta^{-h}}(h) \right] \\
& \leq  \frac{v_{diff}(\emptyset)}{2} \cdot 10^{-P}\calA_1^{max}  + \max_{a \in \calA(h)} L(h \cdot a) \text{.}
\end{align*}}

We will now provide an explicit bound on the loss in the root node $L(\emptyset)$. We have shown that in order to prove the worst case bound it suffices to consider deterministic choice of action at every node --- this means that a single path in the game tree is pursued during propagation of loss. The loss is increased exclusively in imperfect recall nodes and we can encounter at most $d$ such nodes on any path from the root. The increase in such nodes is constant ($[v_{max}(\emptyset)-v_{min}(\emptyset)] \cdot 10^{-P}\calA_1^{max} /2$), therefore the bound is $\epsilon = L(\emptyset) \leq [v_{max}(\emptyset)-v_{min}(\emptyset)] \cdot d \cdot 10^{-P}\calA_1^{max} /2$.

We now know that the expected value of the strategy we have found lies within the interval $[ v^* - \epsilon, v^* ]$, where $v^*$ is the optimal value of the Upper Bound MILP. As $v^*$ is an upper bound on the solution of the original bilinear program, no strategy can be better than $v^*$ --- which means that the strategy we found is $\epsilon$-optimal.
\end{proof}

\begin{theorem}
Let $P_{max}(I_1,a)$ be the maximum number of digits of precision used for representing variable $x(a)$ set as
{\small
\[ P_{max}(I_1,a) = \left\lceil \max_{h \in I_1} \log_{10} \frac{|\calA(I_1)| \cdot d \cdot v_{diff}(h)}{2 \epsilon} \right\rceil \text{,} \]}
where $v_{diff}(h) = v_{max}(h) - v_{min}(h)$.
With this setting, Algorithm~\ref{alg:bnb} terminates and it is guaranteed to return an $\epsilon$-optimal strategy for player 1.
\end{theorem}
\begin{proof}
We start by proving that Algorithm~\ref{alg:bnb} with this choice of $P_{max}(I_1,a)$ terminates. We will show that every branch of the branch-and-bound search tree is finite. This together with the fact that every node is visited at most once and the branching factor of the search tree is finite (every node of the search tree has at most 3 child nodes) ensures that the algorithm terminates.

Every node of the search tree is tied to branching on some variable $x(a)$. Let $p$ be the current precision used to represent $x(a)$ and let us consider the first node on the branch where $x(a)$ is represented with such precision. At such point, $p-1$ digits are fixed and thus $x \in [ c, c + 10^{-(p-1)}]$ for some $c \in [0,1]$. On line~\ref{alg:bnb:refine} an interval of size $10^{-p}$ is handled, every left/right operation (lines~\ref{alg:bnb:left} and~\ref{alg:bnb:right}) may thus handle an interval whose size is reduced at least by $10^{-p}$. We can conduct at most 9 left/right branching operations (lines~\ref{alg:bnb:left} and~\ref{alg:bnb:right}) before the size of the interval drops below $10^{-p}$, which forces us to increase $p$. At most 10 operations can be performed on every $x(a)$ for every precision $p$, the limit on $p$ is finite for every such variable and the number of variables is finite as well, the branch has therefore to terminate.

Let us now show that these limits on the number of refinements $P_{max}(I_1,a)$ are enough to guarantee $\epsilon$-optimality. We will refer the reader to the proof of Theorem~\ref{lemma:alg-bound} for details while we focus exclusively on the behavior in nodes from imperfect recall information sets.

Let $I_1 \in \calI_1^{IR}$ and $h \in I_1$. We know that the L1 distance between behavioral strategies in $I_1$ is at most $10^{-P_{max}(I_1,a)} \cdot |\calA(I_1)|$ (for any $a \in \calA(I_1)$). This means that the bound on $L(h)$ in $h$ from the proof of Theorem~\ref{lemma:alg-bound} is modified to:
{\small
\begin{align*}
L&(h) = v_{\beta^0}(h) - v_{\beta^h}(h) \\
&= \left[ v_{\beta^{-h}}(h) - v_{\beta^h}(h) \right] + \left[ v_{\beta^0}(h) - v_{\beta^{-h}}(h) \right] \\
& \leq \frac{v_{diff}(h)}{2} \cdot 10^{-P_{max}(I_1,a)}\cdot |\calA(I_1)| + \max_{a \in \calA(h)} \!\!L(h \cdot a) \\
& \leq \frac{v_{diff}(h)}{2}  \cdot \frac{|\calA(I_1)|\cdot 2\epsilon}{|\calA(I_1)| \cdot d \cdot [v_{diff}(h)]} + \max_{a \in \calA(h)}\!\!L(h \cdot a) \\
 &= \frac{\epsilon}{d} + \max_{a \in \calA(h)} L(h \cdot a) \text{.}
\end{align*}}

Similarly with the reasoning in the proof of Theorem~\ref{lemma:alg-bound}, it suffices to assume players choosing action at every node in a deterministic way. The path induced by these choices contains at most $d$ imperfect recall nodes, thus $L(\emptyset)=d \cdot \epsilon / d = \epsilon$.
\end{proof}
\begin{theorem}
When using $P_{max}(I_1,a)$ from Theorem \ref{th:prec} for all $I_1 \in \calI_1$ and all $a \in \calA(I_1)$, the number of iterations of the \textsc{BnB} algorithm needed to find an $\epsilon$-optimal solution is in $O(3^{4S_1( \log_{10}(S_1 \cdot v_{diff}(\emptyset)) + 1)} 2^{-5S_1}\epsilon^{-5S_1})$, where $S_1 = |\calI_1|\calA_1^{max}$.
\end{theorem}
\begin{proof}
We start by proving that there is $N \in O(3^{4|\calI_1|\calA_1^{max}P_{max}})$ nodes in the BnB tree, where ${\calA}_1^{max} = \max_{I \in {\calI}_1}|{\calA}(I)|$ and $P_{max} = \max_{I \in {\calI}_1, a \in A(I)}P_{max}(I, a)$. This holds since in the worst case we branch for every action in every information set (hence $|{\calI}_1|{\calA}_1$). We can bound the number of branchings for a fixed action by $4P_{max}$, since there are 10 digits (we branch at most 4 times using binary halving) and we might require $P_{max}$ number of digits of precision. ${4|{\calI}_1|{\calA}_1^{max}P_{max}}$ is therefore the maximum depth of the branch-and-bound tree. Finally the branching factor of the branch-and-bound tree is at most 3.

By substituting
{\small $$\displaystyle\max_{I_1 \in \calI_1}\left\lceil \max_{h \in I_1} \log_{10} \frac{|\calA(I_1)| \cdot d \cdot v_{diff}(h)}{2 \epsilon} \right\rceil$$ }
for $P_{max}$ in the above bound (Theorem \ref{th:prec}), we obtain
{\small
\begin{align*}
N &\in  O(3^{4S_1 \max_{I_1 \in \calI_1}\left\lceil \max_{h \in I_1} \log_{10} \frac{|\calA(I_1)| \cdot d \cdot v_{diff}(h)}{2 \epsilon} \right\rceil}),\\
\shortintertext{\centering\text{where } $S_1 =  |{\calI}_1|{\calA}_1^{max}$}
&\in  O(3^{4S_1\max_{I_1 \in \calI_1}\left\lceil \log_{10} \frac{|\calA(I_1)| \cdot d \cdot v_{diff}(\emptyset)}{2 \epsilon} \right\rceil})\\
&\in  O(3^{4S_1\max_{I_1 \in \calI_1}\left\lceil \log_{10} \frac{S_1 v_{diff}(\emptyset)}{2 \epsilon} \right\rceil})\\
&\in  O(3^{4S_1\left\lceil \log_{10} \frac{S_1 \cdot v_{diff}(\emptyset)}{2 \epsilon} \right\rceil})\\
&\in  O(3^{4S_1( \log_{10} \frac{S_1 \cdot v_{diff}(\emptyset)}{2 \epsilon} + 1)})\\
&\in  O(3^{4S_1( \log_{10}(S_1 \cdot v_{diff}(\emptyset)) - \log_{10}(2 \epsilon) + 1)})\\
&\in  O(3^{4S_1( \log_{10}S_1 \cdot v_{diff}(\emptyset)) + 1)} 3^{-10S_1 \log_{10}(2 \epsilon)})\\
&\in  O(3^{4S_1( \log_{10}S_1 \cdot v_{diff}(\emptyset)) + 1)} 3^{-10S_1\frac{\log_3(2 \epsilon)}{\log_3(10)}})\\
&\in  O(3^{4S_1( \log_{10}S_1 \cdot v_{diff}(\emptyset)) + 1)} (2\epsilon)^\frac{-10S_1}{\log_3(10)})\\
&\in  O(3^{4S_1( \log_{10}(S_1 \cdot v_{diff}(\emptyset)) + 1)} (2\epsilon)^{-5S_1})
%&\in  O(3^{4S_1( \log_{10}(S_1 \cdot v_{diff}(\emptyset)) + 1)} 2^\frac{-10S_1}{\log_3(10)}\epsilon^\frac{-10|{\calI}_1|{\calA}_1^{max}}{\log_3(10)})
\end{align*}}
\end{proof}